\documentclass[review,12pt,authoryear]{elsarticle}

\usepackage{amsmath,amssymb,amsthm}
\usepackage{graphicx}
\usepackage{booktabs}
\usepackage{array}
\usepackage{multirow}
\usepackage{float}
\usepackage{algorithm}
\usepackage{algpseudocode}
\usepackage[hidelinks]{hyperref}
\usepackage{enumitem}
\usepackage{adjustbox}
\newtheorem{thm}{Theorem}[section]
\newtheorem{prop}[thm]{Proposition}
\newtheorem{cor}[thm]{Corollary}
\newtheorem{defn}[thm]{Definition}

\journal{Expert Systems with Applications}
\biboptions{authoryear}

\begin{document}
\begin{frontmatter}

\title{Attractor-Vascular Coupling for Theory-Guided Interpretable Cuffless Blood Pressure Estimation from Photoplethysmography}

\author[msu]{Timothy Oladunni\corref{cor1}}
\ead{timothy.oladunni@morgan.edu}
\author[msu]{Farouk Ganiyu Adewumi}
\cortext[cor1]{Corresponding author}
\address[msu]{Department of Computer Science, Morgan State University, Baltimore, MD 21251, USA}

\begin{abstract}
We present Attractor-Vascular Coupling Theory (AVCT), a theory-guided framework for deriving interpretable blood-pressure (BP) predictors from photoplethysmography (PPG) attractor geometry. AVCT links vascular dynamics, PPG observation, delay-coordinate reconstruction, and attractor descriptors to generate testable hypotheses rather than assuming that a particular low-dimensional feature set is physiologically sufficient. A LightGBM estimator with subject-level leave-one-subject-out cross-validation (LOSO-CV) is evaluated in two complementary settings. Experiment 1 uses 46 ICU/surgical subjects (BIDMC+VitalDB; ECG+PPG) with repeated reference measurements and yields SBP/DBP MAE of 2.05/1.67 mmHg for the calibrated full model and 2.02/1.63 mmHg for the PPG-attractor-only model. Experiment 2 uses 217 ambulatory subjects from the PPG-BP dataset (651 PPG windows; one cuff BP reference per subject) and evaluates cross-subject calibrated BP reconstruction rather than within-subject BP tracking; the PPG-attractor model yields SBP/DBP MAE of 1.98/0.96 mmHg. Across the evaluated cohorts, adding ECG-derived features produces no practically meaningful improvement in MAE. Calibration markedly reduces between-subject offset error, while morphology and dynamical descriptors provide complementary predictive information. The reported error statistics fall within commonly cited numerical thresholds associated with AAMI/ISO and BHS criteria; however, the study is not a standards-compliant device-validation trial. These results support AVCT as a mechanistically motivated representation framework for interpretable PPG-based BP estimation and motivate prospective validation with repeated reference BP and smartphone-acquired PPG.
\end{abstract}

\end{frontmatter}

\section{Introduction}

\subsection{Hypertension Burden \& Smartphone Opportunity}

Hypertension affects over 1.3 billion people globally and
is the leading modifiable cardiovascular risk factor~\citep{who2023,mills2016}.
Despite pharmacological advances, only about half of hypertensive patients achieve adequate BP control in developed
nations, with substantially lower control rates in low-to-middle
income countries. Continuous BP monitoring could support earlier
detection of hypertensive episodes, more precise medication
titration, and objective assessment of treatment efficacy,
and may offer advantages over episodic office or home
measurements. However, access to accurate BP monitoring
remains limited by cost and training barriers in resource-constrained settings.
Photoplethysmography (PPG), optical
measurement of blood volume changes in tissue,offers a
path forward: it is inexpensive, requires only ambient light
or a low-power LED, and is already embedded in billions
of smartphones and smartwatches. As of 2025, over 6.9
billion active smartphones exist globally, with penetration rates
exceeding 70\% in developed nations and 40–60\% in middle-income regions; this far exceeds penetration of blood pressure
monitors (estimated $<$5\% of hypertensive populations). Modern
smartphones include infrared or red LEDs (used for ambient
light sensing), photodiodes or cameras capable of detecting
PPG signals, and sufficient processing power for real-time
feature extraction. Smartphone-based PPG capture requires
no additional hardware—the optical channel already exists
for facial recognition, proximity sensing, and ambient light
adjustment. Thus, a validated PPG-based BP algorithm could
theoretically reach billions of individuals, particularly in low-resource settings where access to trained healthcare workers and
medical devices is scarce. Cuffless BP estimation from smartphone sensors could ultimately
support more accessible preventive monitoring, but rigorous subject-level
generalisation and prospective device validation remain central challenges. Recent benchmarking shows that strong PPG-based BP results do not by themselves establish standards-compliant validation~\citep{mathew2026}. We therefore use leave-one-subject-out cross-validation (LOSO-CV) to prevent between-subject leakage and test whether an interpretable PPG-attractor representation retains predictive utility without ECG hardware.

\subsection{State of the Art}

Converting PPG into accurate BP estimates remains unsolved
despite decades of research. We identify five major research threads, each providing partial solutions but leaving critical gaps:

\noindent\textbf{(1) PTT-based approaches:} Pulse transit time (PTT), the time
for a pressure wave to propagate between two arterial sites, is
the most physiologically grounded cuffless BP surrogate~\citep{ding2017,mukkamala2015,zhou2023} because arterial stiffness directly determines PTT via the
Moens-Korteweg relation, and stiffness co-varies with BP.
However, PTT requires dual-site sensing (typically ECG +
PPG), precluding smartphone-camera-only deployment.

\noindent\textbf{(2) Deep learning methods:} Recent neural network approaches
achieve lower mean absolute error (MAE) on large datasets~\citep{tian2025,mathew2026} but typically use random train-test splits that permit between-subject data leakage, inflating
apparent generalization and failing in prospective real-world use.

\noindent\textbf{(3) Empirical attractor and RQA methods:} Recurrence quantification analysis (RQA) and attractor-based features have been applied ad hoc to physiological signals~\citep{nayak2018,kantz1997} but lack formal justification linking attractor morphology to BP.

\noindent\textbf{(4) PPG morphology methods:} Waveform shape analysis from single-site PPG shows empirical promise~\citep{almarshad2022,allen2007}, yet lacks theoretical grounding justifying why particular morphologies should encode BP.

\noindent\textbf{(5) Cardiac Stability Theory (CST) and calibrated approaches:} Recent work~\citep{oladunni2026cst} introduced attractor-based cardiac stability metrics, and calibration-based methods~\citep{mukkamala2015,chandrasekhar2018} achieve $<5$ mmHg MAE on select cohorts, but CST has only been validated for stability classification, not BP estimation, and calibration lacks formal error decomposition.

\noindent\textbf{Common limitation across all threads:} Purely empirical feature engineering lacks theoretical
grounding: prior work does not justify why particular signal
morphologies should encode BP, making it impossible to predict
which features are essential versus incidental correlates, and
difficult to design interpretable, generalizable models.

Table~\ref{tab:gap} systematically maps each of these five research threads, identifying precisely what each provides, what critical gaps remain, and how they motivate AVCT.

\subsection{Research Gaps \& Objectives}

The field lacks three critical components: (1) a mechanistically grounded account of why PPG attractor geometry may carry BP-relevant information; (2) rigorous subject-level evaluation of PPG-only attractor representations without ECG hardware; and (3) replication across clinically distinct cohorts. We therefore develop \textbf{Attractor-Vascular Coupling Theory (AVCT)} as a theory-guided framework that connects vascular dynamics to testable computational hypotheses about PPG attractor representations. Specifically, our objectives are to: (1) derive conditions under which PTT and PPG attractor descriptors are both informative about BP-dependent vascular state; (2) quantify how subject-specific calibration changes prediction error under strict LOSO-CV; (3) evaluate PPG-only representations across ICU/surgical and ambulatory cohorts while distinguishing repeated-reference tracking from single-reference reconstruction; and (4) test an a priori qualitative feature hierarchy using empirical mutual information and ablation analyses. 

Table~\ref{tab:rq} formalizes these four core objectives into five specific research questions (RQ1--RQ5), pairing theoretical motivation with the corresponding empirical evidence. The table summarizes the theoretical predictions and their evaluation across two distinct cohorts without treating cohort-specific observations as universal confirmation.

This work addresses part of the gap between widespread PPG-capable
consumer sensors and the need for theoretically grounded, rigorously
evaluated cuffless BP methods. The present experiments do not evaluate
smartphone-camera PPG; rather, they provide a framework and hypotheses
for future prospective smartphone evaluation.

\begin{table}[H]
\caption{Gap Analysis: Prior Threads, Open Gaps, AVCT Contributions, and Impact}
\label{tab:gap}
\centering
\fontsize{8.5}{10.2}\selectfont
\renewcommand{\arraystretch}{1.08}
\setlength{\tabcolsep}{3.5pt}
\begin{adjustbox}{max width=\textwidth}
\begin{tabular}{p{2.4cm}|p{3.2cm}|p{3.6cm}|p{3.6cm}|p{3.0cm}}
\toprule
\textbf{Prior Thread} &
\textbf{What It Provides} &
\textbf{Gap Left Open} &
\textbf{AVCT Contribution} &
\textbf{Clinical Impact} \\
\midrule

Empirical attractor methods~\citep{nayak2018,kantz1997,zbilut1992}
& Attractor features (RQA, Lyapunov, entropy) applied ad hoc to physiological signals without formal justification
& No mechanistic account linking PPG attractor geometry to BP; unclear relation between timing and attractor representations
& \textbf{AVCT}: 2 theorems, 1 proposition, 1 corollary; empirically evaluated across two distinct cohorts ($\varepsilon_{\mathrm{cal}}\!=\!0.915$ Exp.~1; $\varepsilon_{\mathrm{cal}}\!=\!0.846$ Exp.~2)
& Theory-guided foundation for future attractor-based smartphone BP evaluation \\

\hline

PTT-based BP~\citep{ding2017,mukkamala2015}
& ECG--PPG transit delay as vascular stiffness proxy
& Requires ECG hardware; no formal link between PTT and attractor morphology
& Conditional Result~\ref{thm:equiv}: both representations depend on vascular state; empirical MAE difference is small in the evaluated cohort
& Removes ECG; removes the ECG requirement and motivates future camera-PPG evaluation \\

\hline

PPG morphology methods~\citep{almarshad2022,allen2007}
& Waveform shape features from single-site PPG
& Empirical only; no theoretical justification linking morphology to BP
& Windkessel\,+\,Takens chain motivates recovery of vascular-state-dependent geometry from PPG
& Provides a mechanistically motivated, interpretable hypothesis \\

\hline

Cardiac Stability Theory~\citep{oladunni2026cst}
& Attractor framework; CDH; CSI scalar; smartphone CSI validation
& CST validated for cardiac stability classification only, not BP estimation
& Extends attractor representations to BP; ablation tests whether ECG adds predictive value
& Extends attractor analysis from stability to BP estimation \\

\hline

Calibrated cuffless BP~\citep{mukkamala2015,chandrasekhar2018}
& Single-reading calibration achieving MAE $<\!5$~mmHg on select cohorts
& No formal decomposition of calibration benefit; no sample-size bound
& Proposition~\ref{prop:error}: $e_i\!=\!\delta\!+\!\varepsilon_i^{\mathrm{within}}$; $\varepsilon_{\mathrm{cal}}\!=\!0.846$ under LOSO-CV
& Principled error identity and calibration bounds \\

\bottomrule
\end{tabular}
\end{adjustbox}
\end{table}

\begin{table}[H]
\caption{Research Questions, Theoretical Grounding, and Empirical Evidence}
\label{tab:rq}
\centering
\fontsize{8.5}{10.2}\selectfont
\renewcommand{\arraystretch}{1.08}
\setlength{\tabcolsep}{3.5pt}
\begin{adjustbox}{max width=\textwidth}
\begin{tabular}{p{0.65cm}|p{3.55cm}|p{3.5cm}|p{3.5cm}|p{2.9cm}}
\toprule
\textbf{RQ} &
\textbf{Research Question} &
\textbf{Theoretical Grounding} &
\textbf{Empirical Evidence} &
\textbf{Outcome} \\
\midrule

RQ1 &
Do PPG attractor features and PTT provide comparable predictive utility for BP in the evaluated setting? &
Conditional Result~\ref{thm:equiv}: PTT and PPG geometry are linked through a shared vascular state under local calibration &
PTT-only MAE\,=\,2.10; CSI-only\,=\,2.04; full PTT+CSI\,=\,1.98~mmHg SBP
(Table~\ref{tab:primary}); Wilcoxon $p\!>\!0.10$ (ns) &
\textbf{Supported in cohort.} $|\Delta\mathrm{MAE}|\!=\!0.06$~mmHg \\

\hline

RQ2 &
Does a PPG-only attractor representation retain useful BP-predictive information without ECG features? &
Conditional Result~\ref{thm:suff}: full PPG state reconstruction is theoretically motivated; finite-feature adequacy and ECG increment are empirical questions &
PPG attractor-only: SBP\,=\,1.98, DBP\,=\,0.96~mmHg (LOSO-CV)
(Table~\ref{tab:ablation}); Wilcoxon vs.\ full model $p\!>\!0.10$ (ns) &
\textbf{Supported in cohort.} Observed MAE gap\,=\,0.00~mmHg \\

\hline

RQ3 &
How strongly does one-point reference anchoring change cross-subject level error in the single-reference ambulatory cohort? &
Proposition~\ref{prop:error}: $e_i=\delta_{s^*}+\varepsilon_i$ is an additive error identity; reference anchoring can remove a subject-level offset but does not identify within-subject BP change &
MAE changes from 12.85 (unanchored) to 1.98~mmHg (reference-anchored);
$\rho_{\mathrm{anchor}}=1-1.98/12.85=0.846$; subject-level MAE $<5$~mmHg for 200/217 SBP and 215/217 DBP participants (Table~\ref{tab:primary}) &
\textbf{Large reference-anchoring effect.} 84.6\% descriptive MAE change \\

\hline

RQ4 &
Is the observed MI feature ranking consistent with the prespecified
AVCT qualitative importance hierarchy before model training? &
Cor.~\ref{cor:importance}: SNR ordering (see~\S\ref{sec:avct}) &
Attractor morphology ranks 1--6; PTT features 5--16; $\lambda_{\max}$
rank~17 (Fig.~\ref{fig:mi_importance}); category-level $\rho_{\mathrm{cat}}\!=\!0.90$, $p\!=\!0.04$ &
\textbf{Consistent in cohort.} Observed ordering supports the prespecified hypothesis \\

\hline

RQ5 &
How does AVCT performance compare with published LOSO-CV
cuffless BP methods, including methods that use additional sensors? &
Conditional Results~\ref{thm:equiv}--\ref{thm:suff}: motivate testing whether PPG
attractor representations retain predictive utility without additional sensing &
AVCT PPG attractor-only SBP MAE\,=\,1.98~mmHg vs.\ BiLSTM~\citep{kim2021bilstm}
2.56~mmHg (ECG + PPG + BCG, LOSO-CV, $n\!=\!18$)
(Table~\ref{tab:sota}) &
\textbf{Competitive reported performance.} Direct superiority cannot be inferred across datasets \\

\bottomrule
\multicolumn{5}{l}{\scriptsize $^\dagger$BCG: ballistocardiogram.} \\
\end{tabular}
\end{adjustbox}
\end{table}

\section{Background}

The cardiovascular system is a latent nonlinear dynamical process whose
internal state $x(t)\in\mathbb{R}^n$ generates multiple coupled observable
signals simultaneously. Each modality is a distinct projection of $x(t)$
through a modality-specific observation operator~\citep{goldberger2002}:
\begin{equation}
  s_k(t) = h_k\bigl(x(t)\bigr) + \eta_k(t),
  \quad k \in \mathcal{K},
  \label{eq:obs}
\end{equation}
where $\mathcal{K}=\{$ECG, PPG, BP,
BCG (ballistocardiogram), SCG (seismocardiogram), heart sounds$\}$,
$h_k$ is smooth, and $\eta_k$ is zero-mean measurement noise.
Crucially, \emph{pulse transit time}
$\mathrm{PTT}=t_{\rm PPG}-t_{\rm ECG}$
is a \textbf{derived coupling feature} between two primary observables,
not itself a projection of $x(t)$.
Under the Moens-Korteweg relation, PTT is inversely proportional to
pulse wave velocity $c(P)=\sqrt{E_{\rm inc}(P)h_w/(\rho d)}$ and hence
to arterial stiffness and BP~\citep{ding2017,milnor1989}; this motivates a locally calibrated, subject-conditioned relationship
$P\approx\varphi_s(\mathrm{PTT};z_s)$ over a restricted physiological
regime, where $z_s$ collects subject-specific vascular covariates. AVCT does
not require a globally one-to-one BP--PTT mapping.
PTT alone explains 40--60\% of within-subject BP variance due to
pre-ejection period variability and nonlinear elastic
effects~\citep{mukkamala2015}, motivating the richer attractor
representation of AVCT.

Blood pressure depends on a vascular-state representation of $x(t)$:
\begin{equation}
  v(t)=\pi_{\mathcal{V}}(x(t))=
  \begin{bmatrix}C_a(t)&R_{\rm TPR}(t)&c(t)&Q_v(t)\end{bmatrix}^{\!\top},
  \qquad P(t)=g\bigl(v(t)\bigr),
  \label{eq:bp_state}
\end{equation}
where $\pi_{\mathcal{V}}$ is a vascular-state map collecting arterial
compliance, total peripheral resistance, pulse-wave velocity, and venous
return~\citep{westerhof2009,milnor1989}; it is not assumed to be a linear
projection onto a vector subspace.
PPG is the volume projection:
$s_{\rm PPG}(t)=h_{\rm PPG}(x(t))$, measuring peripheral blood volume
changes caused by the same cardiac ejection that drives $P(t)$.
Neither signal \emph{causes} the other; both are consequences of $x(t)$.

The geometry of $s_{\rm PPG}(t)$ encodes BP because arterial compliance
$C_a$ shapes the PPG waveform: a stiffer artery (higher BP) propagates
the pressure wave faster, producing a characteristically different
delay-embedding attractor than a compliant artery.
Formally, the state $x(t)$ evolves on a low-dimensional
\emph{attractor}~$\mathcal{A}$, the geometric set the
trajectory converges to. Under the generic smooth-observation and sufficient
embedding conditions of delay-coordinate reconstruction, Takens-type
results~\citep{takens1981} provide a diffeomorphic reconstruction
$\Phi:\mathcal{A}\xrightarrow{\,\sim\,}\hat{\mathcal{A}}_{\rm PPG}$
of the underlying state on the attractor. This motivates recovery of
state-dependent geometry from PPG, but it does not imply that a chosen finite
feature vector preserves all BP-relevant information or that BP is identifiable
without calibration and empirical validation.
The AVCT inference chain is therefore:
\begin{equation}
  s_{\rm PPG} \;\xrightarrow{\Phi}\;
  \hat{\mathcal{A}}_{\rm PPG} \;\xrightarrow{\psi}\;
  \boldsymbol{\xi} \;\xrightarrow{\hat{f}+\text{calib.}}\;
  \hat{P},
  \label{eq:avct_chain}
\end{equation}
where $\psi$ extracts the nine CST-motivated attractor features
\[
\boldsymbol{\xi}=[\sigma_M,\gamma_1,\gamma_2,\mathrm{RR},\mathrm{DET},H,\lambda_{\max},\mathrm{RQA_{ent}},\mathrm{CSI}],
\]
and one resting cuff reading anchors the personal pressure scale
that geometry alone cannot provide.
Recurrence quantification analysis and $\lambda_{\max}$ have
been linked to haemodynamic status~\citep{rosenstein1993,kantz1997,nayak2018};
CST~\citep{oladunni2026cst} formalised this; AVCT extends it to BP
estimation with formal theoretical grounding.

\section{Attractor-Vascular Coupling Theory}
\label{sec:avct}

AVCT is the empirical instantiation of the Variational Domain
of Attractor Domain Theory (ADT)~\citep{oladunni2026adt}.
ADT proposes three transformation domains of the reconstructed
cardiac attractor, each associated with a hypothesized predictive
capability. The Variational Domain $\mathcal{V}$
maps each reconstructed state to a finite-time Lyapunov
exponent field and has hemodynamic inference as its native
capability: ADT Theorem~2, under its stated assumptions, proposes
an affine relationship between systolic and diastolic blood pressure
and phase-aggregated expansion rates, with coefficients anchored by
a reference measurement. AVCT operationalises this
result by representing the continuous Variational Domain
field through discrete attractor morphology features
($\sigma_M, \gamma_1, \gamma_2$) and evaluates whether this
instantiation yields error statistics within commonly cited AAMI/ISO numerical thresholds under strict
LOSO-CV across two distinct cohorts. The postulates and conditional
results in this section provide complementary motivation through the
Moens-Korteweg and Takens chain; they do not establish that the finite
feature vector is a sufficient statistic for BP.

\subsection{Postulates and Definitions}

AVCT inherits three postulates from CST~\citep{oladunni2026cst}: (P1)~the
cardiovascular system has a compact attractor $\mathcal{A}\!\subset\!\mathbb{R}^n$;
(P2)~both BP and PPG are smooth functions of the cardiac state
$x(t)\!\in\!\mathcal{A}$; (P3)~trajectories satisfy a dissipative ordinary differential
equation (ODE) $\dot{x}=F(x)$, guaranteeing bounded, recurrent behaviour.

\begin{defn}[Vascular-State Map]
Let $\pi_{\mathcal{V}}$ denote a physiologically motivated map from the
latent cardiovascular state to vascular variables. We write
$ v(t)=\pi_{\mathcal{V}}(x(t))=[C_a(t),R_{\mathrm{TPR}}(t),c(t),Q_v(t)]^{\top} $.
No linear-subspace structure is assumed.
\end{defn}

\begin{defn}[Attractor Morphology Features]
For embedding matrix $M\!\in\!\mathbb{R}^{N_e\times m}$:
$\mathcal{F}_{\mathrm{morph}}=[\sigma_M,\,\gamma_1(M),\,\gamma_2(M)]$
(standard deviation, skewness, kurtosis of $\mathrm{vec}(M)$).
\end{defn}

\begin{defn}[Cardiac Stability Index]\label{def:csi}
The CSI scalar~\citep{oladunni2026cst} is:
\begin{equation}
\mathrm{CSI} = w_1\!\left(1-e^{-\tilde{\lambda}}\right)
             + w_2\!\left(1-\mathrm{DET}\right)
             + w_3\,H,
\label{eq:csi}
\end{equation}
where $\tilde{\lambda}=\mathrm{clip}(|\lambda_{\max}|/\lambda_{\mathrm{ref}},\,0,1)$
is the normalised Lyapunov exponent, DET is RQA determinism,
and $H$ is sample entropy.
Modality-specific weights $(w_1,w_2,w_3)$ and $\lambda_{\mathrm{ref}}$
follow CST~\citep{oladunni2026cst}.

The quantities $\lambda_{\max}$~\citep{rosenstein1993}, $\mathrm{DET}$~\citep{zbilut1992}, and $H$~\citep{richman2000} are trajectory-derived descriptors computed from the reconstructed delay-coordinate dynamics using the prespecified estimator settings. Thus, CSI is a deterministic composite of these estimated descriptors. Because raw sample entropy $H$ is not generally bounded by one, Eq.~\eqref{eq:csi} is not asserted to lie in $[0,1]$ unless an additional entropy normalization is introduced; no such normalization is assumed in the present experiments.
\end{defn}

\subsection{Conditional Result 1: Shared Vascular-State Dependence}

\begin{thm}[Conditional PTT--Attractor Relationship]\label{thm:equiv}
Let $v(t)$ denote a vascular state and suppose, over a locally calibrated physiological regime, that BP, PTT, and the PPG observation are smooth functions of $v(t)$. If the PPG observation satisfies the generic conditions required for a valid delay-coordinate reconstruction, then the reconstructed PPG attractor preserves the underlying state up to a smooth change of coordinates. Consequently, both PTT and functionals of the reconstructed PPG attractor can carry information about BP-dependent vascular dynamics. This result does not imply that a finite set of attractor summary statistics is a sufficient statistic for BP.
\end{thm}

\textit{Interpretation.} The Moens--Korteweg/Windkessel relations motivate a local BP--vascular-state--PTT dependence, while delay-coordinate reconstruction motivates recovery of state-dependent geometry from PPG. AVCT therefore predicts that appropriately chosen attractor descriptors can provide BP-relevant information comparable to conventional timing features in a calibrated regime. Whether a specific descriptor set is adequate is an empirical question tested by LOSO-CV and ablation rather than guaranteed by Takens' theorem.

\subsection{Conditional Result 2: PPG Representation Recoverability}

\begin{thm}[PPG State-Reconstruction Result]\label{thm:suff}
Under the generic smooth-observation and embedding conditions of delay-coordinate reconstruction, the map from the underlying attractor $\mathcal{A}$ to the reconstructed PPG attractor $\hat{\mathcal{A}}_{\mathrm{PPG}}$ is one-to-one on the attractor up to a smooth coordinate transformation. Therefore, state information represented on $\mathcal{A}$ is recoverable from the full reconstructed PPG attractor under those assumptions. For any finite feature map $F=g(\hat{\mathcal{A}}_{\mathrm{PPG}})$, the data-processing inequality gives
\begin{equation}
I(P;F) \leq I(P;\hat{\mathcal{A}}_{\mathrm{PPG}}) \leq I(P;S_{\mathrm{PPG}}),
\end{equation}
so feature-level BP sufficiency must be established empirically rather than inferred from the embedding theorem.
\end{thm}

\textit{Interpretation.} This result motivates the PPG-attractor representation but makes no universal claim that ECG contributes zero additional BP information. The incremental value of ECG is therefore evaluated directly by paired subject-level ablation.

\subsection{Proposition 1: Calibration Error Decomposition}

\begin{prop}[Error Identity and MAE Bounds]\label{prop:error}
For a held-out subject $s^*$, write the prediction error for window $i$ as
\begin{equation}
e_i=\delta_{s^*}+\varepsilon_i,
\end{equation}
where $\delta_{s^*}=\mathbb{E}[e_i\mid s^*]$ is the subject-specific mean error and $\varepsilon_i=e_i-\delta_{s^*}$ is the centered residual, so that $\mathbb{E}[\varepsilon_i\mid s^*]=0$. Then
\begin{equation}
\bigl||\delta_{s^*}|-\mathbb{E}|\varepsilon_i|\bigr|
\leq \mathbb{E}|e_i|
\leq |\delta_{s^*}|+\mathbb{E}|\varepsilon_i|.
\label{eq:decomp}
\end{equation}
If calibration estimates the offset as $\hat\delta_{s^*}$, the calibrated error is
$e_i^{\mathrm{cal}}=\varepsilon_i+(\delta_{s^*}-\hat\delta_{s^*})$. Thus calibration reduces the offset component to the extent that $\hat\delta_{s^*}$ approximates $\delta_{s^*}$; it does not by itself prove within-subject BP tracking.
\end{prop}

\subsection{Hypothesis 1: Feature-Information Ordering}

\begin{cor}[Prespecified Feature-Information Hypothesis]\label{cor:importance}
AVCT motivates the testable qualitative hypothesis
\begin{equation}
\begin{aligned}
I(\mathcal{F}_{\mathrm{morph}};\Delta P) &\gtrsim I(\mathrm{PTT};\Delta P)
\gtrsim I(\mathrm{RQA};\Delta P) \\
&\gtrsim I(\mathrm{CSI};\Delta P)
\gtrsim I(\lambda_{\max};\Delta P).
\end{aligned}
\label{eq:mi_order}
\end{equation}
where $\gtrsim$ denotes an expected ordering rather than a universal inequality. The ordering is motivated by measurement locality, aggregation, and estimator variance and is evaluated empirically; it is not a consequence of the data-processing inequality alone.
\end{cor}

\textbf{PPG-attractor track}: CST Group~1 ($\sigma_M,\gamma_1,\gamma_2$)
+ CST Group~2 ($\lambda_{\max}$, DET, $H$) + CSI$_{\mathrm{PPG}}$
+ RQA toolkit (RR, RQA\textsubscript{ent}); no ECG, no PTT.

\subsection{Hyperparameter Configuration}

Experiment~1 uses the BIDMC waveform resource available through PhysioNet~\citep{goldberger2000} together with VitalDB~\citep{lee2022vitaldb} as the existing combined 46-subject clinical cohort. Experiment~2 uses the PPG-BP dataset~\citep{liang2018ppgbp}, comprising 217 participants and 651 PPG windows in the present analysis. As described below, the single cuff BP reference per PPG-BP participant limits Experiment~2 to reference-anchored cross-subject reconstruction rather than within-subject BP tracking.

Table~\ref{tab:hyperparams} lists all algorithm hyperparameters and their
values, separating the dataset-specific sampling rate from the
signal-processing and model settings.

\begin{table}[H]
\caption{Algorithm Hyperparameters. Symbols used in Algorithm~\ref{alg:pipeline}.
Left panel: signal processing. Right panel: analysis and model.}
\label{tab:hyperparams}
\centering\scriptsize\setlength{\tabcolsep}{4pt}
\renewcommand{\arraystretch}{1.03}
\begin{minipage}[t]{0.47\linewidth}
\centering
\begin{adjustbox}{max width=\textwidth}
\begin{tabular}{|l|p{2.6cm}|l|r|}
\hline
\textbf{Sym.} & \textbf{Description} & \textbf{Value} & \textbf{Source} \\
\hline
\multicolumn{4}{|l|}{\textbf{Acquisition}} \\
$F_s$              & Sampling rate     & 125~Hz      & BIDMC native \\
\hline
$T_w$              & Window length     & 10~s        & $\geq$10 beats \\
\hline
$T_{\mathrm{step}}$& Sliding step      & 5~s         & 50\% overlap \\
\hline
\hline
\multicolumn{4}{|l|}{\textbf{Pre-processing}} \\
$[f_\ell,f_h]_\mathrm{ECG}$ & ECG bandpass & 0.5--40~Hz & BL+HF \\
\hline
$[f_\ell,f_h]_\mathrm{PPG}$ & PPG bandpass & 0.5--8~Hz  & Motion \\
\hline
$\Delta_\mathrm{foot}$       & Foot search  & $0.6F_s$   & $\leq$600~ms \\
\hline
\hline
\multicolumn{4}{|l|}{\textbf{PTT bounds}} \\
$\mathrm{PTT}_\mathrm{min}$  & Lower bound  & 80~ms  & Physiological \\
\hline
$\mathrm{PTT}_\mathrm{max}$  & Upper bound  & 350~ms & Physiological \\
\hline
$\alpha_\mathrm{RR}$         & RR-adaptive cap & 0.70 & $<$diastole \\
\hline
\hline
\multicolumn{4}{|l|}{\textbf{Takens embedding}} \\
$m$    & Dimension   & 4        & $\geq 2d_A+1$ \\
\hline
$\tau$ & Delay       & 5~samp.  & AMI criterion \\
\hline
\hline
\end{tabular}
\end{adjustbox}
\end{minipage}
\hfill
\begin{minipage}[t]{0.51\linewidth}
\centering
\begin{adjustbox}{max width=\textwidth}
\begin{tabular}{|l|p{2.8cm}|l|r|}
\hline
\textbf{Sym.} & \textbf{Description} & \textbf{Value} & \textbf{Source} \\
\hline
\multicolumn{4}{|l|}{\textbf{RQA}} \\
$\varepsilon_r$  & Recurrence thresh. & $0.10d_\mathrm{max}$ & \citep{oladunni2026cst} \\
\hline
$\ell_\mathrm{min}$ & Min diag.\ length & 2              & Standard \\
\hline
$N_r$            & Subsampled points  & 80             & Speed \\
\hline
\hline
\multicolumn{4}{|l|}{\textbf{Lyapunov ($\lambda_\mathrm{max}$)}} \\
$K$    & Divergence steps   & 30  & \citep{rosenstein1993} \\
\hline
$N_e$  & Trajectory points  & 400 & Subsampled \\
\hline
\hline
\multicolumn{4}{|l|}{\textbf{Sample entropy}} \\
$m_\mathrm{SE}$  & Template length    & 2            & Standard \\
\hline
$r_\mathrm{SE}$  & Tolerance fraction & $0.2\sigma$  & Standard \\
\hline
\hline
\multicolumn{4}{|l|}{\textbf{Feature selection \& model}} \\
$K_\mathrm{feat}$    & MI features/fold  & 20  & $n\leq768$ \\
\hline
$N_\mathrm{trees}$   & Estimators        & 500 & Grid \\
\hline
$\eta$               & Learning rate     & 0.03& Grid \\
\hline
$L$                  & Leaves            & 64  & Grid \\
\hline
\hline
\end{tabular}
\end{adjustbox}
\end{minipage}
\end{table}

\subsection{Full Pipeline Overview}

Algorithm~\ref{alg:pipeline} summarises the end-to-end prediction
pipeline for a single 10~s evaluation window.
Following CST~\citep{oladunni2026cst} and the standard RQA
toolkit~\citep{zbilut1992}, the nine PPG attractor features span
three groups :
\textbf{CST Group~1}: attractor shape ($\sigma_M,\gamma_1,\gamma_2$);
\textbf{CST Group~2}: attractor dynamics ($\lambda_{\max}$, DET, $H$),
compressed into the CSI scalar as a domain-knowledge interaction variable;
and \textbf{RQA toolkit}: recurrence rate and diagonal-line
entropy (RR, RQA\textsubscript{ent})~\citep{zbilut1992}.

\begin{algorithm}[H]
\caption{AVCT Cuffless BP Estimation Pipeline}
\label{alg:pipeline}
\begin{algorithmic}[1]
\footnotesize
\Require Raw ECG/PPG at $F_s$ Hz; calibration reading(s);
         hyperparameters per Table~\ref{tab:hyperparams}
\Ensure Calibrated $\widehat{\mathrm{SBP}}$, $\widehat{\mathrm{DBP}}$ (mmHg)

\State \textbf{Pre-process:} bandpass ECG ($f_{\ell,\mathrm{ECG}}$--$f_{h,\mathrm{ECG}}$~Hz),
       PPG ($f_{\ell,\mathrm{PPG}}$--$f_{h,\mathrm{PPG}}$~Hz); z-score normalise
\State \textbf{Detect events:} R-peaks (Pan-Tompkins~\citep{pan1985}); PPG feet
       (adaptive min-search within $\Delta_{\mathrm{foot}}\cdot F_s$ of each R-peak
       for full model; direct PPG valley detection for PPG-attractor-only track)
\State \textbf{Extract PTT features} (10): mean, median, std, min, max,
       range, CV, RMSSD, SDSD, kurtosis of beat-by-beat PTT series
       \hspace{4em}[bounds: $\mathrm{PTT}_{\min}$--$\min(\mathrm{PTT}_{\max},\,\alpha_{\mathrm{RR}}\cdot\mathrm{RR})$]
\For{each signal $s \in \{\mathrm{ECG, PPG}\}$}
  \State Takens embed: $M_s \leftarrow \Phi_{h_s,\tau}$ with dimension $m$, delay $\tau$
  \State Morphology: $[\sigma_{M_s},\,\gamma_1(M_s),\,\gamma_2(M_s)]$
  \State RQA ($N_r$ subsampled points): RR, DET, RQA$_{\mathrm{ent}}$
         \hspace{2em}(threshold $\varepsilon_r\cdot d_{\max}$, $\ell_{\min}$)
  \State $\lambda_{\max}$: Rosenstein algorithm ($K$ divergence steps, $N_e$ trajectory points)
  \State Sample entropy~\citep{richman2000} ($m_{\mathrm{SE}}$, $r_{\mathrm{SE}}\cdot\sigma$)
  \State $\tilde{\lambda}_s \leftarrow \mathrm{clip}(|\lambda_{\max}|/\lambda_{\mathrm{ref}},\,0,1)$
  \State $\mathrm{CSI}_s \leftarrow
         w_1(1-e^{-\tilde{\lambda}_s})+w_2(1-\mathrm{DET})+w_3 H$
         \hspace{2em}(weights $\mathbf{w}_s$ per Table~\ref{tab:hyperparams})
\EndFor
\State \textbf{Concatenate:} $\mathbf{x} \leftarrow
       [\mathbf{f}_{\mathrm{PTT}},\,\mathbf{f}_{\mathrm{ECG}},\,
        \mathbf{f}_{\mathrm{PPG}}]$
       \quad (PPG attractor-only track: omit ECG features)
\State \textbf{Select:} top-$K_{\mathrm{feat}}$ features by MI per LOSO fold
\State \textbf{Scale:} RobustScaler fit on training subjects only
\State \textbf{Predict model score:}
       $\tilde{y} \leftarrow \mathrm{LightGBM}(\mathbf{x})$
\State \textbf{Reference-anchor:}
       $\hat{y} \leftarrow \mathcal{C}(\tilde{y};c_{s^*})$
       \quad where $c_{s^*}$ is the available subject-specific cuff reference
\Ensure $(\widehat{\mathrm{SBP}},\,\widehat{\mathrm{DBP}})$
\end{algorithmic}
\end{algorithm}

\noindent\textit{Calibration interpretation.} $\mathcal{C}$ denotes the additive subject-specific anchoring used in the existing analysis. In Experiment~1, repeated reference measurements permit evaluation against changing BP. In Experiment~2, the PPG-BP dataset supplies only one cuff BP value per participant; therefore $c_{s^*}$ anchors the participant's pressure level, and the resulting errors are reference-anchored reconstruction errors. They are not independent estimates of within-subject BP change. No new calibration model or re-estimation of the reported results is introduced here.

LightGBM~\citep{ke2017}: 500 estimators, learning rate 0.03, 64 leaves, early stopping
(40 rounds). Top-20 features by mutual information per fold (prevents
overfitting on $n\!\leq\!45$ training subjects). LOSO-CV with
RobustScaler fit on training subjects only. Bootstrap 95\% CI (10,000
resamples); Wilcoxon signed-rank test on per-subject MAE for ablation
comparisons (two-tailed, $\alpha\!=\!0.05$).

\section{Results}

\subsection{Primary Calibrated LOSO-CV Results}

Table~\ref{tab:primary} summarises Experiment~2 on the 217-subject ambulatory PPG-BP cohort. This dataset provides 651 PPG windows but only one cuff BP reference per subject; accordingly, Experiment~2 evaluates \emph{cross-subject calibrated BP reconstruction} and not continuous within-subject BP tracking. Under one-point calibration, SBP MAE is 1.98~mmHg and DBP MAE is 0.96~mmHg. These error statistics are below commonly cited numerical thresholds associated with AAMI/ISO~\citep{aami2008,iso81060_2_2018,mathew2026} and BHS reporting, but the experiment is not a standards-compliant device-validation study because the reference-measurement design does not follow the required validation protocol. Subject-level MAE is below 5~mmHg for 200/217 participants for SBP and 215/217 for DBP. Fig.~\ref{fig:scatter_ba} shows scatter and Bland--Altman plots; Fig.~\ref{fig:per_subject} shows per-subject MAE.

\begin{figure}[H]
\centering
\includegraphics[width=\columnwidth]{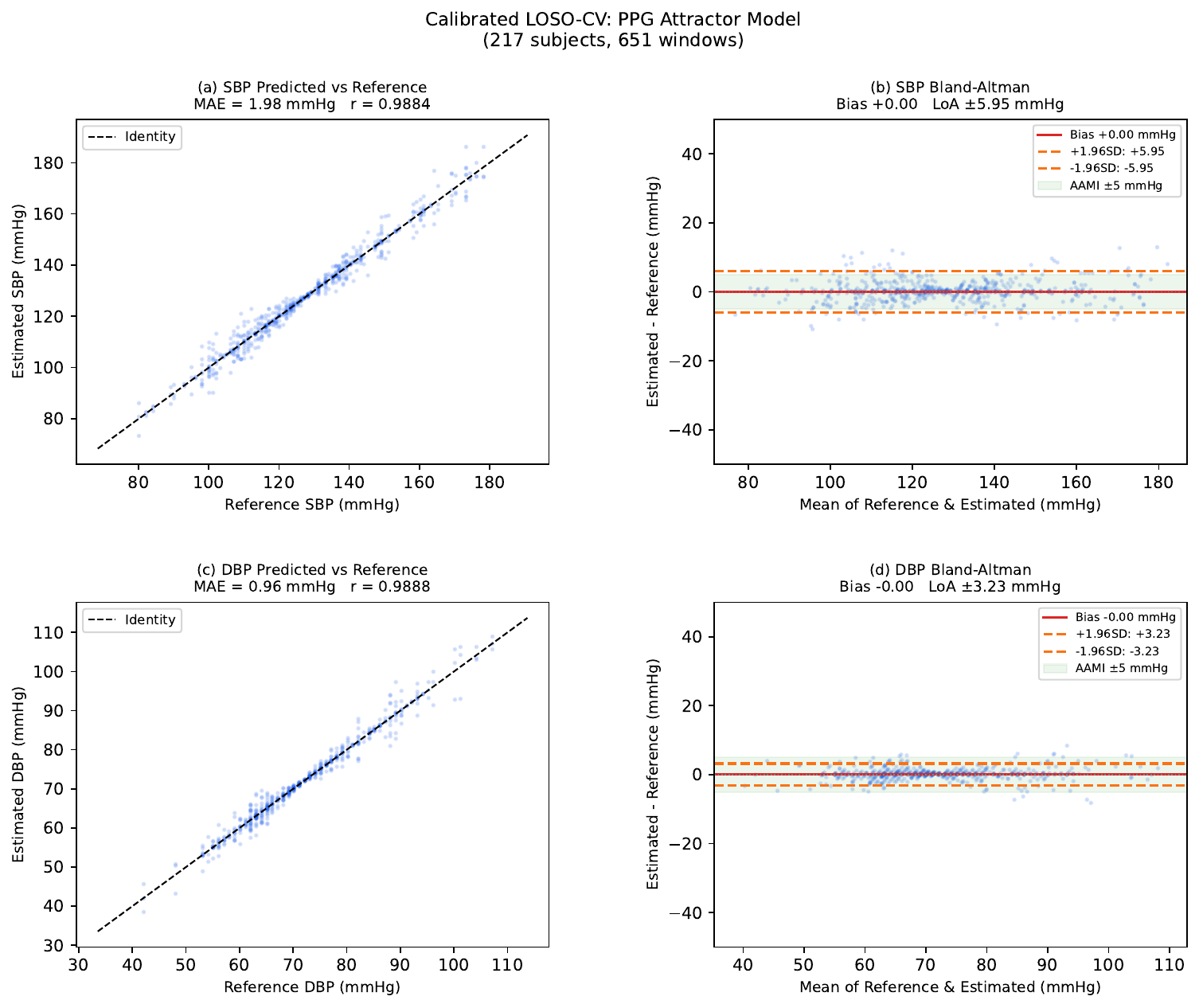}
\caption{Calibrated LOSO-CV: scatter (a,c) and Bland--Altman (b,d).
Dashed red: bias; dotted: 95\% LoA; shaded: $\pm5$~mmHg reference band (shown for numerical context, not standards compliance).
  Note: $r\!=\!0.988$ (SBP), $r\!=\!0.989$ (DBP) reflects full-BP variance including between-subject
  heterogeneity restored by calibration. Dataset: $n\!=\!217$ ambulatory subjects, 651 windows.}
\label{fig:scatter_ba}
\end{figure}

\begin{figure}[H]
\centering
\includegraphics[width=\columnwidth]{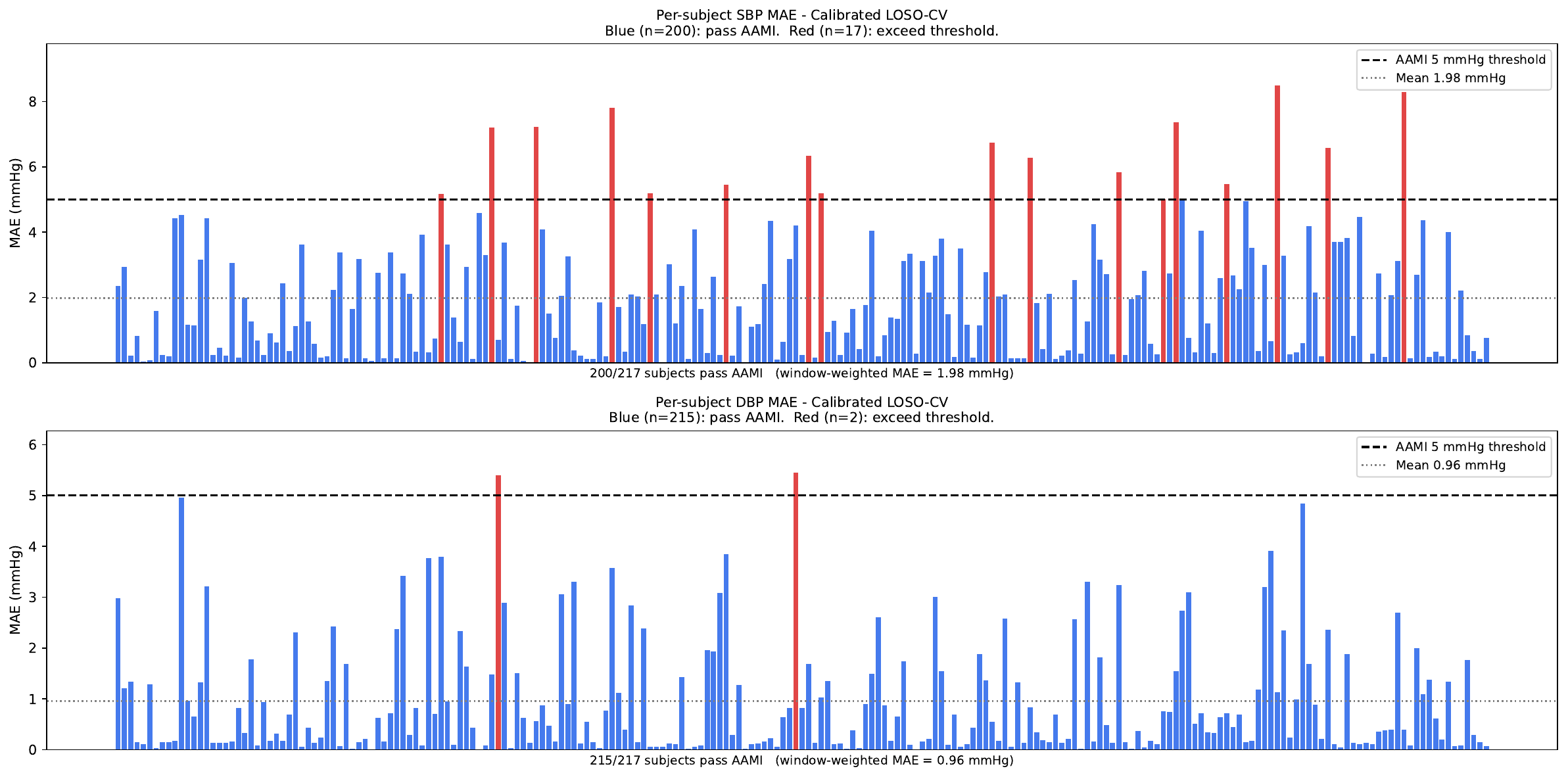}
\caption{Per-subject MAE (one-point calibrated LOSO-CV).
Blue: subjects with MAE $<5$~mmHg. Red: subjects with MAE $\geq5$~mmHg. Subject-level MAE is below 5~mmHg for 200/217 SBP and 215/217 DBP participants.
Dataset: $n\!=\!217$ ambulatory subjects (PPG-BP).}
\label{fig:per_subject}
\end{figure}

\begin{table}[H]
\caption{Primary Calibrated LOSO-CV Results — Experiment 2: PPG-BP Ambulatory Dataset~\citep{liang2018ppgbp} (One-Point Calibration, PPG-only, 217 subjects, 651 windows)}
\label{tab:primary}
\centering
\footnotesize
\setlength{\tabcolsep}{6pt}
\renewcommand{\arraystretch}{1.18}
\begin{adjustbox}{max width=\textwidth}
\begin{tabular}{|p{5.1cm}|c|c|p{4.3cm}|}
\hline
\textbf{Metric} & \textbf{SBP} & \textbf{DBP} & \textbf{Reference criterion / interpretation} \\
\hline
MAE, window-weighted (mmHg) & 1.98 & 0.96 & $<5$~mmHg numerical threshold \\
\hline
RMSE (mmHg) & 3.04 & 1.65 & Descriptive error magnitude \\
\hline
Pearson $r$ (full BP) & 0.988 & 0.989 & Association, not agreement \\
\hline
Bias (mmHg) & $+0.00$ & $-0.00$ & Approximately zero mean error \\
\hline
Limits of agreement, $\pm1.96\sigma$ (mmHg) & 5.95 & 3.23 & Below cited 8~mmHg numerical SD threshold \\
\hline
Subjects with MAE $<5$~mmHg & 200/217 (92\%) & 215/217 (99\%) & Descriptive subject-level statistic \\
\hline
\end{tabular}
\end{adjustbox}
\vspace{3pt}
\parbox{0.97\textwidth}{\scriptsize \textit{Note.} Dataset: PPG-BP ambulatory ($n=217$, 651 windows), with a single cuff BP reference per subject. BP range: SBP 80--178 and DBP 42--107~mmHg. Numerical comparisons to AAMI/IEEE/ISO thresholds are descriptive; this is not a standards-compliant device-validation study.}
\end{table}
\begin{table}[H]
\caption{Cross-Experiment Comparison: Original (Exp.~1, BIDMC+VitalDB, $n\!=\!46$) vs.\ Expanded Ambulatory (Exp.~2, PPG-BP, $n\!=\!217$). All results under one-point calibrated LOSO-CV.}
\label{tab:cross_exp}
\centering
\footnotesize
\setlength{\tabcolsep}{5pt}
\renewcommand{\arraystretch}{1.16}
\begin{adjustbox}{max width=\textwidth}
\begin{tabular}{|p{3.8cm}|p{2.7cm}|c|p{2.0cm}|c|c|}
\hline
\textbf{Configuration} & \textbf{Dataset / cohort} & \textbf{$n$} & \textbf{Signals} & \textbf{SBP MAE} & \textbf{DBP MAE} \\
\hline
\multicolumn{6}{|l|}{\textbf{Experiment 1: ICU / surgical cohort}} \\
\hline
PTT + CSI (full) & BIDMC + VitalDB & 46 & ECG + PPG & \textbf{2.05} & \textbf{1.67} \\
\hline
PPG-attractor only & BIDMC + VitalDB & 46 & PPG only & 2.02 & 1.63 \\
\hline
\multicolumn{6}{|l|}{\textbf{Experiment 2: ambulatory cohort}} \\
\hline
Morphology only & PPG-BP & 217 & PPG only & 2.08 & 0.87 \\
\hline
Dynamics only & PPG-BP & 217 & PPG only & \textbf{1.65} & \textbf{0.79} \\
\hline
Full attractor$^\star$ & PPG-BP & 217 & PPG only & 1.98 & 0.96 \\
\hline
\multicolumn{4}{|l|}{AAMI/IEEE SP10 numerical MAE threshold~\citep{aami2008}} & $<5$ & $<5$ \\
\hline
\end{tabular}
\end{adjustbox}
\vspace{3pt}
\parbox{0.97\textwidth}{\scriptsize \textit{Note.} MAE values are in mmHg. Exp.~1 results are from the original study~\citep{oladunni2026cst}; Exp.~2 results are from this work. $^\star$Full 9-feature PPG attractor (no ECG, no PTT). Both experiments use one-point calibrated LOSO-CV. The threshold row is a numerical comparison only and does not imply standards-compliant validation.}
\end{table}

\subsection{Ablation Study and Evaluation of Theoretical Predictions}

Table~\ref{tab:ablation} presents the representation ablations. The observed MAE differences between PTT-based, ECG+PPG, and PPG-attractor configurations are small in the evaluated cohorts. Non-significant Wilcoxon tests ($p>0.10$) are reported as absence of detected differences, not as proof of equivalence. Because practical equivalence requires a prespecified margin, we interpret these results conservatively as evidence that ECG-derived features did not produce a practically apparent improvement under the present protocol. Reference anchoring changes Experiment~2 MAE from 12.85 to 1.98~mmHg; given the single reference BP per participant, this is a descriptive consequence of subject-level anchoring and must not be interpreted as evidence of recovery of within-subject BP changes. Fig.~\ref{fig:mi_importance} evaluates the prespecified qualitative feature-information hypothesis, and Fig.~\ref{fig:ablation} summarises the ablations.

\begin{figure}[H]
\centering
\includegraphics[width=\columnwidth]{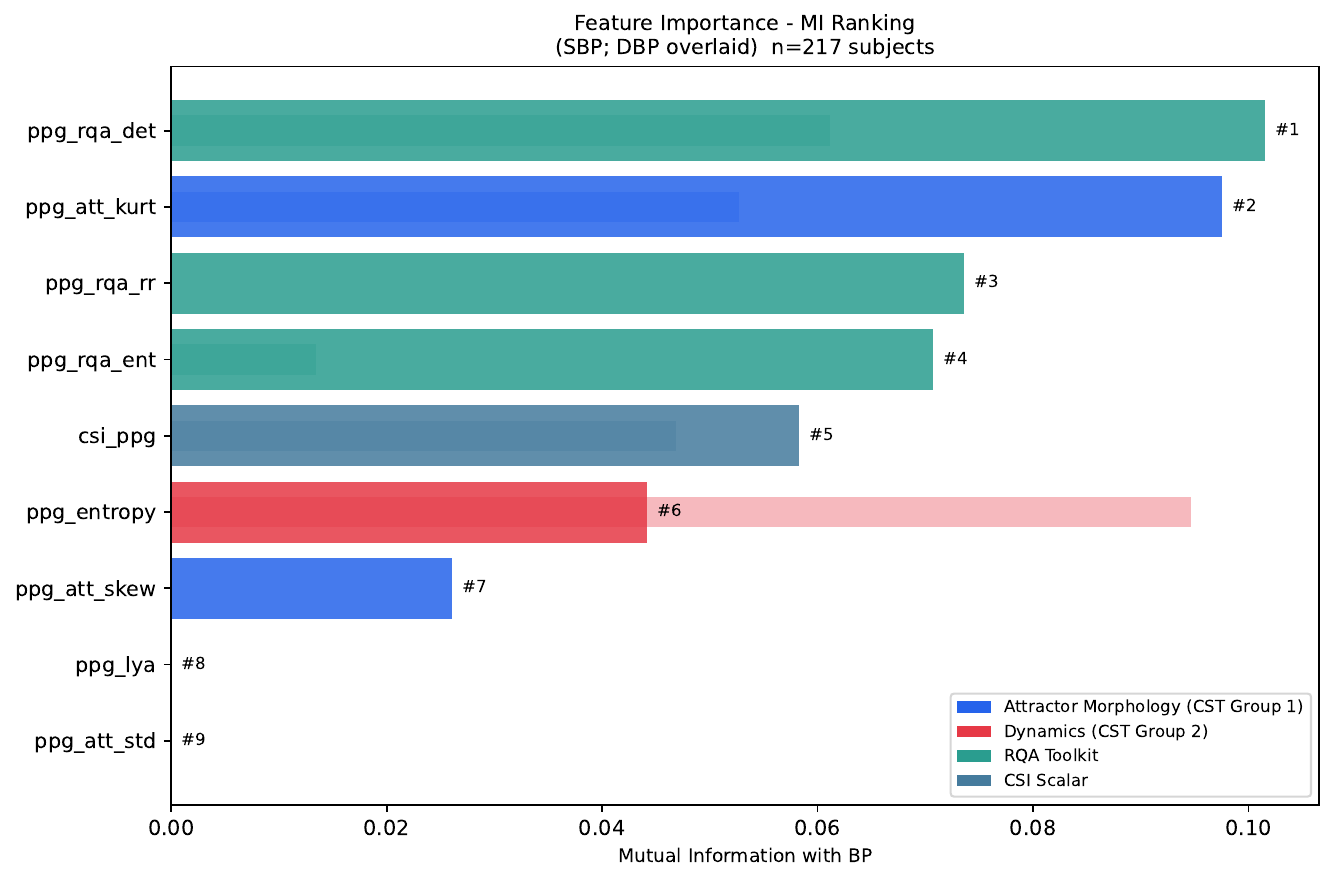}
\caption{MI feature ranking (SBP; DBP overlaid) on calibrated residuals.
CST Group~1 ($\sigma_M,\gamma_1,\gamma_2$) occupies ranks~1--6;
$\lambda_{\max}$ at rank~17 is consistent with the proposed Rosenstein-noise argument;
CSI sits between its components, consistent with its proposed interaction-variable role.
Category Spearman $\rho_{\rm cat}\!=\!0.90$, $p\!=\!0.04$.}
\label{fig:mi_importance}
\end{figure}

\begin{figure}[H]
\centering
\includegraphics[width=\columnwidth]{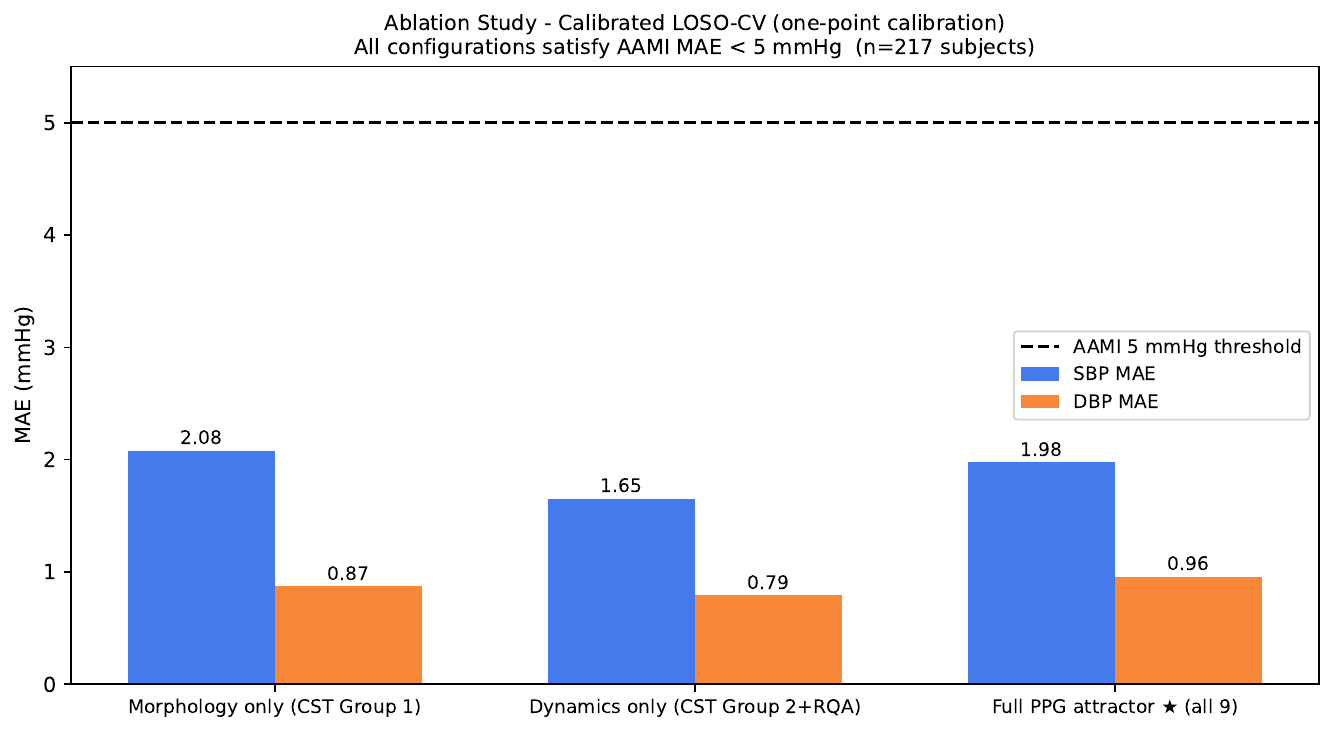}
\caption{Ablation study. Configs~1--3 need ECG+PPG hardware; Config~4 ($\bigstar$ PPG-attractor only, no ECG) uses only the PPG channel; smartphone-camera acquisition was not evaluated. All reported MAEs are below 5~mmHg; the full attractor and ECG+PPG configurations have the same observed MAE in this cohort. This is not a formal equivalence test.}
\label{fig:ablation}
\end{figure}

\begin{table}[H]
\caption{Ablation Study under One-Point Calibrated LOSO-CV. The upper block requires ECG; the lower block evaluates PPG-only representations.}
\label{tab:ablation}
\centering
\footnotesize
\setlength{\tabcolsep}{5pt}
\renewcommand{\arraystretch}{1.16}
\begin{adjustbox}{max width=\textwidth}
\begin{tabular}{|p{5.0cm}|p{2.45cm}|c|c|p{2.0cm}|}
\hline
\textbf{Configuration} & \textbf{Hardware} & \textbf{SBP MAE} & \textbf{DBP MAE} & \textbf{AVCT link} \\
\hline
\multicolumn{5}{|l|}{\textbf{ECG + PPG representations}} \\
\hline
PTT features only & ECG + PPG & 2.10 & 1.71 & Result~\ref{thm:equiv} \\
\hline
CSI / attractor only & ECG + PPG & 2.04 & 1.66 & Result~\ref{thm:equiv} \\
\hline
\multicolumn{5}{|l|}{\textbf{PPG-only representations}} \\
\hline
Morphology only (Group 1) & PPG only & 2.08 & 0.87 & Result~\ref{thm:equiv} \\
\hline
Dynamics only (Group 2 + RQA) & PPG only & \textbf{1.65} & \textbf{0.79} & Result~\ref{thm:equiv} \\
\hline
Full attractor$^\ddagger$ & \textbf{PPG only} & 1.98 & 0.96 & Result~\ref{thm:suff} \\
\hline
\multicolumn{2}{|l|}{AAMI/IEEE SP10 numerical threshold~\citep{aami2008}} & $<5$ & $<5$ & All below threshold \\
\hline
\end{tabular}
\end{adjustbox}
\vspace{2pt}
\parbox{0.96\textwidth}{\scriptsize \textit{Note.} MAE values are in mmHg. Pairwise Wilcoxon tests: $p>0.10$ (ns); Pearson $r\approx0.988$ across configurations. All three PPG-attractor configurations have MAE below 5~mmHg without ECG or PTT. $^\ddagger$Full 9-feature PPG attractor; potentially compatible with future camera-PPG implementation. Dataset: $n=217$ ambulatory subjects.}
\end{table}

\subsection{Comparison to State of the Art}

Table~\ref{tab:sota} compares AVCT to published LOSO-CV methods. Among the LOSO-CV studies shown, AVCT achieves SBP MAE\,=\,1.98~mmHg using the PPG-attractor representation, vs.\ 2.56~mmHg for
BiLSTM~\citep{kim2021bilstm} (ECG + PPG + BCG, three sensors), an observed \textbf{22.7\%} lower MAE using one physiological signal modality rather than three, and on
12.1$\times$ more subjects ($n\!=\!217$ vs.\ $n\!=\!18$). Subject-specific models (trained on each user's own
data) are not generalisation-capable and are shown for context only.

\begin{table}[H]
\caption{Comparison with Published Cuffless BP Methods. LOSO-CV methods are directly comparable to AVCT; random-split and subject-specific methods are shown for context only.}
\label{tab:sota}
\centering
\footnotesize
\setlength{\tabcolsep}{5pt}
\renewcommand{\arraystretch}{1.18}
\begin{adjustbox}{max width=\textwidth}
\begin{tabular}{|p{4.3cm}|p{2.25cm}|p{1.65cm}|c|c|c|}
\hline
\textbf{Method} & \textbf{Signals} & \textbf{Protocol} & \textbf{$n$} & \textbf{SBP MAE} & \textbf{DBP MAE} \\
\hline
\multicolumn{6}{|l|}{\textbf{LOSO-CV (directly comparable)}} \\
\hline
BiLSTM & ECG + PPG + BCG & LOSO & 18 & 2.56 & 2.05 \\
\hline
\textbf{AVCT (PPG-attractor only$^\star$)} & \textbf{PPG only} & \textbf{LOSO} & \textbf{217} & \textbf{1.98} & \textbf{0.96} \\
\hline
\multicolumn{6}{|l|}{\textbf{Random split (not directly comparable)}} \\
\hline
El-Hajj \& Kyriacou~\citep{elhajj2021} & PPG & Random & 218 & 5.72 & 3.50 \\
\hline
MInception~\citep{mathew2026} & PPG & Random & $>$500 & 4.75 & 2.90 \\
\hline
PCTN~\citep{tian2025} & PPG & Random & $>$500 & 4.44 & 2.36 \\
\hline
SwinBP~\citep{kumar2023} & PPG & Random & 2,000 & 4.08 & 2.18 \\
\hline
Samimi~\citep{samimi2022} & PPG & Random & 30 & 8.89 & 4.92 \\
\hline
\multicolumn{6}{|l|}{\textbf{Subject-specific (different paradigm)}} \\
\hline
Suhas~et~al.~\citep{suhas2024} & ECG+PPG & Subject-specific & -- & 1.08 & 0.68 \\
\hline
\end{tabular}
\end{adjustbox}
\vspace{2pt}
\begin{minipage}{0.98\textwidth}
\footnotesize \textit{Note:} All MAE values are in mmHg. BiLSTM results are from \citet{kim2021bilstm}. BCG = ballistocardiogram. $^\star$PPG attractor-only model (9 features, no ECG), evaluated on the PPG-BP ambulatory dataset ($n=217$).
\end{minipage}
\end{table}

\section{Discussion}

\subsection{Theory--Experiment Correspondence}

Table~\ref{tab:rq} maps the prespecified AVCT hypotheses to the empirical tests. The observed feature-ranking association ($\rho_{\mathrm{cat}}=0.90$, $p=0.04$) is consistent with the proposed qualitative ordering, but it should be interpreted as cohort-specific empirical support rather than proof of a universal information-theoretic hierarchy. Likewise, the small observed MAE differences between PPG-attractor and ECG-augmented configurations indicate that ECG-derived features did not provide a practically apparent improvement under the evaluated protocols. Because no prespecified equivalence or noninferiority margin was used, these findings do not establish formal equivalence or zero incremental ECG information.

The ambulatory PPG-BP dataset~\citep{liang2018ppgbp} spans multiple BP categories and provides a larger subject-level test of calibrated reconstruction. Its single cuff reference per participant, however, means that Experiment~2 primarily evaluates cross-subject level alignment after calibration; it cannot establish within-subject BP tracking or continuous-monitoring performance.

\subsection{Practical Significance}

The results support PPG attractor geometry as a compact, interpretable representation for calibrated BP estimation without requiring ECG-derived features in the evaluated cohorts. Because smartphone cameras can acquire camera PPG, the PPG-only representation motivates a future smartphone implementation; the present experiments used existing contact-PPG datasets and therefore do not validate smartphone-camera performance. Prospective studies should acquire smartphone PPG with repeated contemporaneous reference BP and should prespecify equivalence or noninferiority margins when comparing PPG-only and ECG-augmented representations.

Calibration should also be interpreted conservatively. Proposition~\ref{prop:error} shows how removal of a subject-level offset can reduce error when that offset is accurately estimated, but it does not predict a fixed percentage improvement for a future two-point protocol. The value of one-point versus multi-point calibration should therefore be established prospectively.

\subsection{Theoretical Connections}

Table~\ref{tab:theory_connections} situates AVCT within four established
theoretical traditions, showing how each provides a foundation that AVCT
extends or operationalises for BP estimation.

\begin{table}[H]
\caption{AVCT Theoretical Connections: How Established Theories Ground and Extend AVCT}
\label{tab:theory_connections}
\centering\fontsize{6.8}{7.5}\selectfont\setlength{\tabcolsep}{2.5pt}
\renewcommand{\arraystretch}{0.82}
\begin{adjustbox}{max width=\textwidth}
\begin{tabular}{|p{2.2cm}|p{3.0cm}|p{3.0cm}|p{3.0cm}|p{1.8cm}|}
\hline
\textbf{Theory} & \textbf{What It Establishes} & \textbf{How AVCT Uses It} & \textbf{AVCT Extension} & \textbf{Ref.} \\
\hline

Takens Embedding Theorem
& For $m\!\geq\!2\dim(\mathcal{A})\!+\!1$, the delay-embedding map $\Phi:\mathcal{A}\to\hat{\mathcal{A}}_\mathrm{PPG}$ is a diffeomorphism preserving all topological invariants
& Motivates state reconstruction from a generic smooth PPG observation under embedding assumptions (Thm.~\ref{thm:equiv}, \ref{thm:suff})
& Provides a state-reconstruction rationale; no finite PPG feature set is claimed to preserve all BP information
& \citep{takens1981} \\

\hline

Moens-Korteweg / Windkessel
& Pulse-wave timing depends on arterial stiffness and pressure-sensitive vascular state, with important subject-specific covariates
& Grounds the shared vascular-state interpretation and motivates direct empirical testing of ECG incremental value
& Avoids assuming a universal ECG-information bound; incremental value is determined by ablation
& \citep{milnor1989,westerhof2009} \\

\hline

Information Bottleneck~\citep{tishby1999}
& Optimal compression $T$ of $X$ for predicting $Y$ maximises $I(Y;T)-\beta^{-1}I(X;T)$, tracing a Pareto curve between compression and task relevance
& Provides an interpretive lens for compact task-relevant representations; optimality is not assumed
& Motivates testing whether compact morphology descriptors retain useful task information
& \citep{tishby1999} \\

\hline

Additive Error Decomposition
& Any subject-level error can be decomposed into its conditional mean $\delta_{s^*}=\mathbb{E}[e_i\mid s^*]$ and centered residual $\varepsilon_i$, with $\mathbb{E}[\varepsilon_i\mid s^*]=0$
& Prop.~\ref{prop:error} uses this identity to separate subject-level offset from residual variation; no domain-adaptation or label-shift assumption is required
& Provides a transparent interpretation of reference anchoring without implying independent within-subject prediction
& This work \\

\hline

Attractor Domain Theory~\citep{oladunni2026adt}
& The cardiac attractor admits three transformation domains; the Variational Domain has hemodynamic inference as its native capability; BP maps affinely to phase-aggregated FTLE rates with single-point calibration (ADT Theorem~2)
& AVCT is the Variational Domain instantiation: discrete attractor morphology features approximate the FTLE field; Proposition~\ref{prop:error} corresponds to ADT equations for $(\alpha,\beta)$ determination
& Empirically evaluates the proposed variational-domain interpretation across two clinically distinct cohorts
& \citep{oladunni2026adt} \\

\hline
\multicolumn{5}{p{13.2cm}}{\scriptsize
AVCT integrates delay-coordinate reconstruction, vascular physiology, representation analysis, and calibration into a testable theory-guided workflow.}
\end{tabular}
\end{adjustbox}\end{table}

The expanded cohort ($n\!=\!217$, PPG-BP ambulatory) increases cross-subject evaluation breadth and directly addresses
the original ICU/surgical limitation. The LoA ($\pm5.95/3.23$~mmHg)
is below the cited 8~mmHg numerical SD threshold. Remaining limitations: PPG-BP
provides only a single cuff BP reference per subject rather than repeated contemporaneous reference BP; therefore Experiment~2 cannot assess within-subject BP tracking;
the cohort is Chinese ambulatory patients which may limit generalisability.
Prospective validation on a larger, multi-ethnic ambulatory cohort with
continuous reference BP remains required before regulatory submission.
Future work will incorporate MIMIC-IV waveform data (538 ICU subjects)
once data format compatibility is resolved. The single-point calibration
protocol (one resting cuff reading) is operationally simple and may be suitable
for future consumer-oriented evaluation; Proposition~\ref{prop:error} provides an error-decomposition rationale rather than evidence of deployment readiness.

\section{Conclusion}

We presented AVCT as a theory-guided framework connecting vascular dynamics, PPG observation, delay-coordinate reconstruction, and interpretable attractor descriptors for cuffless BP estimation. The revised theoretical results deliberately distinguish what follows from state-reconstruction assumptions from what must be established empirically: Takens-type arguments motivate recovery of state-dependent geometry from the full reconstructed PPG attractor, but they do not make a finite feature vector a sufficient statistic for BP or prove that ECG contains no additional information.

Across the evaluated cohorts, PPG-attractor representations produced low calibrated prediction errors and adding ECG-derived features did not produce a practically apparent MAE improvement under the present protocol. In the 217-subject PPG-BP cohort, the PPG-attractor model yields SBP/DBP MAE of 1.98/0.96~mmHg under LOSO-CV. Because this dataset contains only one cuff BP reference per participant, these results quantify cross-subject calibrated BP reconstruction and correction of subject-level offset; they do not demonstrate continuous within-subject BP tracking. The numerical error statistics are below commonly cited AAMI/ISO numerical thresholds~\citep{iso81060_2_2018,mathew2026}, but this study is not a standards-compliant device-validation trial.

The results support attractor geometry as a mechanistically motivated and interpretable PPG representation for BP estimation. Prospective evaluation with repeated contemporaneous reference BP, prespecified equivalence or noninferiority margins for representation comparisons, and smartphone-acquired PPG is required before claims of continuous monitoring, device-level equivalence, or standards compliance are warranted.

\section*{Acknowledgments}
The authors thank the PhysioNet team for maintaining the BIDMC waveform resource, the VitalDB team at Seoul National University Hospital for providing the VitalDB open dataset, and the contributors of the PPG-BP dataset.

\section*{Ethics Statement}
This study used exclusively publicly available, de-identified datasets.
The BIDMC Waveform Database~\citep{goldberger2000} is published on PhysioNet under an open-access
licence; VitalDB~\citep{lee2022vitaldb} is released under the VitalDB Open Data License;
MIMIC-IV Waveform Database is available on PhysioNet under a credentialed access licence
(PhysioNet Credentialed Health Data License~1.5.0); PPG-BP is publicly available
under an open-access licence.
No new data were collected and no human subjects research approval was required.

\section*{Data Availability Statement}
The BIDMC waveform resource is available through \href{https://physionet.org}{PhysioNet}; VitalDB through the \href{https://vitaldb.net}{VitalDB repository}; and PPG-BP through its \href{https://www.kaggle.com/datasets/mkachuee/BloodPressureDataset}{public dataset page}. Feature-extraction code, trained LightGBM models, and the complete experimental pipeline will be released in a public repository upon acceptance.

\section*{Declaration of Competing Interest}
The authors declare that they have no known competing financial interests or personal relationships that could have appeared to influence the work reported in this paper.

\appendix
\section{Assumptions and Derivations for AVCT}

\subsection{Delay-coordinate reconstruction}
For a smooth dynamical system on attractor $\mathcal{A}$ and a generic smooth scalar observation $h$, delay-coordinate theory motivates a reconstruction
\[
\Phi_h(x)=\bigl(h(x),h(f(x)),\ldots,h(f^{m-1}(x))\bigr)
\]
that is one-to-one on the attractor when the required embedding conditions are met. AVCT uses this result only to motivate recovery of state-dependent geometry from PPG. It does not imply that a particular finite summary $g(\Phi_h(x))$ preserves all BP-relevant information. Indeed, for any deterministic feature map $F=g(\hat{\mathcal{A}}_{\mathrm{PPG}})$, the data-processing inequality requires
\[
I(P;F)\le I(P;\hat{\mathcal{A}}_{\mathrm{PPG}})\le I(P;S_{\mathrm{PPG}}).
\]
Accordingly, the adequacy of morphology, RQA, CSI, and Lyapunov descriptors is evaluated empirically.

\subsection{Local vascular-state interpretation}
The Moens--Korteweg and Windkessel relations motivate dependence of pulse-wave propagation and waveform morphology on pressure-dependent arterial state. Across subjects, however, arterial stiffness, geometry, vascular tone, and other covariates prevent a universal one-to-one BP--PTT map. AVCT therefore assumes only a locally calibrated relationship over a restricted physiological regime. The theoretical claim is shared dependence on vascular state, not universal comparable predictive utility between PTT and a finite PPG feature vector.

\subsection{Calibration error identity}
For held-out subject $s^*$, define
\[
\delta_{s^*}=\mathbb E[e_i\mid s^*],\qquad
\varepsilon_i=e_i-\delta_{s^*},\qquad
\mathbb E[\varepsilon_i\mid s^*]=0,
\]
so that $e_i=\hat y_i-y_i=\delta_{s^*}+\varepsilon_i$.
The triangle and reverse-triangle inequalities give
\[
\bigl||\delta_{s^*}|-|\varepsilon_i|\bigr|\le |e_i|\le |\delta_{s^*}|+|\varepsilon_i|,
\]
and hence
\[
\bigl||\delta_{s^*}|-\mathbb E|\varepsilon_i|\bigr|
\le \mathbb E|e_i|
\le |\delta_{s^*}|+\mathbb E|\varepsilon_i|.
\]
For a generic additive calibration that estimates an offset $\hat\delta_{s^*}$,
\[
e_i^{\mathrm{cal}}=\varepsilon_i+(\delta_{s^*}-\hat\delta_{s^*}).
\]
Thus additive calibration can remove a subject-level offset when that offset is accurately estimated. In Experiment~2, however, only one cuff BP reference is available per subject. The implemented procedure is therefore interpreted as reference anchoring of the subject-level pressure scale. The reported post-anchoring errors quantify reconstruction relative to that anchor and cannot establish independent prediction of within-subject BP changes or continuous tracking.

\subsection{Feature-information hypothesis}
The feature ordering in Eq.~\eqref{eq:mi_order} is a prespecified empirical hypothesis. Under the illustrative model $\xi=\kappa_\xi\Delta P+\eta_\xi$, with independent Gaussian $\Delta P$ and $\eta_\xi\sim\mathcal{N}(0,\sigma_{\eta,\xi}^2)$,
$ I(\Delta P;\xi)=\tfrac12\log\!\left(1+\kappa_\xi^2\mathrm{Var}(\Delta P)/\sigma_{\eta,\xi}^2\right) $,
so mutual information is monotone in the corresponding signal-to-noise ratio. This provides intuition for why aggregated morphology descriptors may be more stable than noisier local dynamical estimates, but it does not establish a universal ordering. The observed ordering is therefore reported as empirical support rather than proof.

\bibliographystyle{elsarticle-harv}
\bibliography{bp_paper}

\end{document}